\documentclass[12pt,a4paper]{article}

\usepackage{amsmath}
\usepackage{amsthm}
\usepackage{amsfonts}
\usepackage{amssymb}
\usepackage{graphicx}
\usepackage{pbox}%%%using this to break a cell in an array so I can get a line break
\usepackage{tikz, tikz-3dplot}
\usetikzlibrary{decorations.markings,arrows}
\usepackage[section]{placeins}%This prevents placing floats before a section.
\usepackage{pifont} % this is just to get nice bullet styles in itemize, \ding{226}
\usepackage[colorinlistoftodos]{todonotes}%%%this gives nice todonotes in preparation. lets see if it works
\usepackage{hyperref}

\hypersetup{
	colorlinks,
	linkcolor={red!50!black},
	citecolor={blue!50!black},
	urlcolor={blue!80!black}
}

\theoremstyle{plain}
\newtheorem{thm}{Theorem}%%going to just number plainly

\newtheorem{propn}{Proposition}
\newtheorem{lem}{Lemma}
\theoremstyle{definition}

\newtheorem{rmk}{Remark}
\newcommand{\R}{\mathbf{R}}

\title{The Maxwell Conjecture is False }
\author{Philip Arathoon\thanks{Babson College, MA, USA, \texttt{parathoon@babson.edu}}, Gavin Ball\thanks{University of Missouri, Columbia, MO, USA, \texttt{gavin.ball@missouri.edu}}, Matthew D. Kvalheim\thanks{University of Maryland, Baltimore County, MD, USA, \texttt{kvalheim@umbc.edu}}}
\date{July 2026}
\begin{document}
	\maketitle
	\begin{abstract}
		
		We exhibit a configuration of five point charges in Euclidean space whose electrostatic potential admits at least 24 critical points all of which are non-degenerate. Maxwell's conjecture that the field of \(n\) point charges has at most \((n-1)^2\) critical points which are all non-degenerate is therefore false.
	\end{abstract}
	
	% \begin{center}
		% \begin{minipage}{0.8\textwidth}
			% 	\fbox{\textbf{AI statement...}
				% 		\vspace{1em}
				
				% Hello	
				% }
			% \end{minipage}
		% \end{center}
	
	\section{Introduction}
	
	In J. C. Maxwell's 1873 treatise on electricity and magnetism he discusses the number of equilibria of the electric field generated by \(n\) point charges \cite[\S113]{maxwell1873treastise}. Apparently unaware of this, M. Morse and S. S. Cairns in 1969 posed the problem of finding an upper bound for the number of equilibria \cite[p.~293]{morse1969critical}. The first general bounds were supplied by A. Gabrielov, D. Novikov, and B. Shapiro in \cite{gabrielov2007mystery} who, based on their reading of \cite[\S113]{maxwell1873treastise}, formulated the `Maxwell conjecture' which states that if the critical points of the electrostatic potential generated by \(n\) point charges are all non-degenerate then their number cannot exceed \((n-1)^2\). These bounds were later improved by V. Zolotov in 2023 \cite{zolotov2023} and further improved by H. Edelsbrunner, C. Fillmore, and G. Oliveira in 2026 \cite{edelsbruner2026counting}. 
	Maxwell's bound is trivially achieved for \(n=2\) but it is not known even for \(n=3\) if 4 is the maximum number, except in the case of equal charges \cite{tsai2015maxwell}. 
	Further related problems in classical electrostatics are discussed in \cite{abanov2021around}. 
	
	In this note we produce a counterexample to Maxwell's conjecture: a configuration of five point charges with at least 24 critical points all of which are non-degenerate. The counterexample starts with three unit charges placed at the vertices of an equilateral triangle. This configuration has four equilibria: one at the centre and three displaced inwards around the edges of the triangle, as shown in Figure~1(a). Two small charges are then added to the centre of the triangle and moved slightly apart along the orthogonal axis to form a shallow triangular bipyramid. The three edge equilibria persist under the addition of the small axial charges whereas the central equilibrium bifurcates into a family of 21 equilibria.

	\subsubsection*{Acknowledgments} Kvalheim was supported in part by the Air Force Office of Scientific Research under award number FA9550-24-1-0299.
	
	\subsubsection*{Tool and computational resource disclosure} The idea behind this construction was suggested by an LLM (OpenAI's GPT-5.6 Sol). The authors have verified the mathematical details and have written the argument in their own words. Computer algebra software (Mathematica, Maple) was used to verify computations and produce visualisations.
	
	\section{Counterexample}
	
	Place three unit charges at the vertices of an equilateral triangle
	\[
	\mathbf{a}_1=\begin{pmatrix}
		1\\ 0\\0
	\end{pmatrix},
	\quad
	\mathbf{a}_2=\begin{pmatrix}
		-1/2\\ \sqrt{3}/2 \\0
	\end{pmatrix},
	\quad
	\mathbf{a}_3=\begin{pmatrix}
		-1/2 \\-\sqrt{3}/2 \\0
	\end{pmatrix}.
	\]
	The Coulomb potential generated by these charges has the following Taylor expansion around the origin
	\[
	V_\triangle=\sum_{j=1}^3\frac{1}{\|\mathbf{x}-\mathbf{a}_j\|}=3+\frac{3}{4}H_2+\frac{15}{8}H_3+\frac{9}{64}H_4+O(\|\mathbf{x}\|^5)
	\]
	where 
	\begin{align*}
		H_2&=r^2-2z^2,\\
		H_3&=x^3-3xy^2,\\
		H_4&=3r^4-24r^2z^2+8z^4,
	\end{align*}
	are harmonic polynomials with \(r^2=x^2+y^2\). Place two additional \(q_\varepsilon\)-charges along the triangle's axis of symmetry at
	\[
	\mathbf{a}_4=\begin{pmatrix}
		0\\0\\\varepsilon
	\end{pmatrix},\quad
	\mathbf{a}_5=\begin{pmatrix}
		0\\0\\-\varepsilon
	\end{pmatrix}.
	\]
	The potential generated by the axial charges has the expansion
	\[
	V_\pm^\varepsilon=\sum_{j=4}^5\frac{q_\varepsilon}{\|\mathbf{x}-\mathbf{a}_j\|}=\frac{2q_\varepsilon}{\varepsilon}-\frac{q_\varepsilon}{\varepsilon^3}H_2+\frac{q_\varepsilon}{4\varepsilon^5}H_4+O\left(\frac{q_\varepsilon\|\mathbf{x}\|^6}{\varepsilon^7}\right).
	\]

	\begin{thm}\label{thm:maintheorem}
		There is an \(\varepsilon_0>0\) such that for all \(0<\varepsilon<\varepsilon_0\) the potential \(V_\varepsilon=V_\triangle+V^\varepsilon_\pm\) generated by five charges located at \(\mathbf{a}_1,\dots,\mathbf{a}_5\) with respective strengths \(1,1,1,q_\varepsilon,q_\varepsilon\) and
		\begin{equation}\label{qe}
			q_\varepsilon=\frac{3}{4} \varepsilon^3 - \frac{5}{32} \varepsilon^5
		\end{equation}
		admits at least 24 non-degenerate critical points. Moreover, a small perturbation to the charge strengths gives a potential with a finite number \(k\ge 24\) of critical points all of which are non-degenerate.
	\end{thm}

	To prove Theorem \ref{thm:maintheorem}, we introduce a deformed potential $\Phi_{\varepsilon},$ defined as
	\[
	\Phi_\varepsilon(\mathbf{X})=\frac{V_\varepsilon(\varepsilon^2\mathbf{X})-V_\varepsilon(\mathbf{0})}{\varepsilon^6},\quad\text{for }\varepsilon>0.
	\]
	% {\color{magenta}For every compact \(K\subset\R^3\) and for \(q_\varepsilon\) in \eqref{qe} there is ...such that...where \(\Phi_0 =...\).}
	\begin{lem}\label{lem:convergence}
		Let $q_{\varepsilon} = \frac{3}{4} \varepsilon^3 - \frac{5}{32} \varepsilon^5$ as in (\ref{qe}) and let
		\[
		\Phi_0 = \frac{5}{32} H_2+\frac{15}{8}H_3+\frac{3}{16}H_4.
		\]
		Then there exists an open neighbourhood $U$ of $\{ 0 \} \times \R^3$ and a function $\Psi$ real-analytic on $U$ such that
		\[
		\Phi_\varepsilon(\mathbf X) = \Phi_0(\mathbf X)+\varepsilon^2\Psi(\varepsilon^2,\mathbf X),
		\]
		whenever $\varepsilon>0$ and $(\varepsilon^2,\mathbf X)\in U.$ Consequently, for every $k\geq0$, and every $K \subset \R^3$ compact we have $\lVert \Phi_\varepsilon-\Phi_0 \rVert_{C^k(K)} \leq C_{K,k}\varepsilon^2$ for all sufficiently small \(\varepsilon>0\).
	\end{lem}
	
	\begin{proof}
		Let $K \subset \R^3$ be any compact set. Write $t = \varepsilon^2.$ Since $V_{\triangle}$ is real analytic near $\mathbf{0}$ and the $H_j$ are homogeneous, we have
		\[
		V_{\triangle}(t \mathbf X) - V_{\triangle}(\mathbf 0) = \frac{3}{4} t^2 H_2 + \frac{15} {8} t^3 H_3 + \frac{9}{64} t^4 H_4 + t^5 R_{\triangle} (t,\mathbf X)
		\]
		where $R_{\triangle}$ is real-analytic near $\{ 0 \} \times K.$ For the axial pair, 
		\[
		V^{\varepsilon}_{\pm}(t \mathbf X) - V^{\varepsilon}_{\pm}(\mathbf 0) = t \left(\frac{3}{4}-\frac{5}{32} t \right) F (\varepsilon,\mathbf X),
		\]
		where
		\[
		F(\varepsilon,\mathbf X)
		=
		\frac1{\sqrt{1-2\varepsilon Z+\varepsilon^2|\mathbf X|^2}}
		+
		\frac1{\sqrt{1+2\varepsilon Z+\varepsilon^2|\mathbf X|^2}}
		-2.
		\]
		This function is real analytic and even in $\varepsilon$ and its Taylor expansion is
		\[
		F(\varepsilon,\mathbf X)
		=
		-\varepsilon^2H_2(\mathbf X)
		+\frac{\varepsilon^4}{4}H_4(\mathbf X)
		+\varepsilon^6R_{\pm}(\varepsilon^2,\mathbf X).
		\]
		Substituting $\varepsilon^2=t$ and adding the $V_\triangle$ and $V^{\varepsilon}_{\pm}$ terms gives
		\[
		V_{\varepsilon}(t \mathbf X) - V_{\varepsilon}(\mathbf 0) = t^3 \left( \frac{5}{32} H_2 + \frac{15}{8} H_3 + \frac{3}{16} H_4 \right) + t^4 R (t,\mathbf X).
		\]
		for some real-analytic $R.$ Dividing by $t^3=\varepsilon^6$ proves the formula for $\Phi_{\varepsilon}$. The $C^k$-estimate follows because the derivatives of $R$ are bounded on a neighbourhood of $\{0\}\times K$.
	\end{proof}

	\begin{rmk}
		The charge \(q_\varepsilon\) is judiciously chosen to have the form in \eqref{qe} for two reasons. Firstly, in order for \(\lim_{\varepsilon\rightarrow 0}\Phi_\varepsilon\) to exist the order of \(q_\varepsilon\) must be at least \(\varepsilon^3\), the coefficient of \(\varepsilon^3\) must equal $3/4,$ and there must be no $\varepsilon^4$ term. Secondly, the specific coefficient of \(\varepsilon^5\) is chosen out of convenience to simplify the form of \(\Phi_0\); it may be replaced with any negative coefficient greater than \(-45/256\) without affecting our results.
	\end{rmk}
	
	\begin{lem}\label{lem:Phi0-critical-oints}
		The polynomial \(\Phi_0\) has exactly 21 critical points all of which are non-degenerate. 
	\end{lem}
	\begin{proof}
		In cylindrical coordinates 
		\[\Phi_0=\frac{1}{32}\left(18 r^4+60 r^3 \cos3\theta -144 r^2 z^2+5 r^2+48 z^4-10 z^2\right)
		\]
		with partial derivatives
		\begin{align*}
			\frac{\partial\Phi_0}{\partial r}&=\frac{r}{16}\left(36 r^2+90 r \cos3 \theta-144 z^2+5\right),\\[.5em]
			\frac{\partial\Phi_0}{\partial \theta}&=-\frac{45r^3}{8}\sin 3\theta,\\[.5em]
			\frac{\partial\Phi_0}{\partial z}&=\frac{z}{8}(48z^2-72r^2-5).
		\end{align*}
		The critical points can now be found and their Hessians evaluated with a routine calculation. These points are classified as follows:
		\begin{itemize}
			\item[\ding{226}] The origin \(r=z=0\) with signature \((++-)\).
			\item[\ding{226}] Two axial equilibria at \(r=0\) and \(z=\pm\frac{\sqrt{15}}{12}\) with signature \((+--)\).
			\item[\ding{226}] Six planar equilibria at \(z=0\), \(\theta=\frac{\pi}{3},\pi,\frac{5\pi}{3}\), and \(r=\frac{15\pm\sqrt{205}}{12}\) with signature \((++-)\) for the outer radius and \((+--)\) for the inner.
			\item[\ding{226}] Twelve off-planar equilibria at \((r,z)=\left(\frac{1}{3},\pm\frac{\sqrt{39}}{12}\right)\) with signature \({(+--)}\) and \((r,z)=\left(\frac{1}{6},\pm\frac{\sqrt{21}}{12}\right)\) with signature \((++-)\) for \(\theta=0,\frac{2\pi}{3},\frac{4\pi}{3}\).
		\end{itemize}
	\end{proof}

	\begin{rmk}
		For the original triangular configuration the central equilibrium has Morse index 1 and the three edge equilibria each have Morse index 2. These satisfy the Euler characteristic equation
		\[
		m_0-m_1+m_2-m_3=0
		\]
		for \(m_1=1\), \(m_2=3\), and where \(m_3=3\) represents the charge locations and \(m_0=1\) the point at infinity. When the two small axial charges are added the three edge equilibria persist but the central equilibrium bifurcates into the critical points of \(\Phi_0\), 10 of which have Morse index 1, and 11 have Morse index 2. The Morse inequalities therefore remain consistent with \(m_0=1\), \(m_1=10\), \(m_2=14\), and \(m_3=5\).
	\end{rmk}

	\begin{proof}[Proof of Theorem~1]
		By Lemma~\ref{lem:convergence}, the function $(\varepsilon,\mathbf{X}) \mapsto \Phi_\varepsilon(\mathbf{X})$
		is smooth, so the implicit function theorem implies that the 21 non-degenerate critical points of $\Phi_0$ provided by Lemma~\ref{lem:Phi0-critical-oints} persist to 21 non-degenerate critical points of $\Phi_\varepsilon$ for all sufficiently small $\varepsilon > 0$.
		These yield 21 non-degenerate critical points of $V_\varepsilon$  at distance $O(\varepsilon^2)$ from the origin.
		
		On the other hand, the implicit function theorem also implies that the three non-zero non-degenerate critical points of $V_\triangle$ persist to three non-degenerate critical points of $V_\varepsilon$ for all sufficiently small $\varepsilon >0$. The distance of these from the origin is bounded away from $0$, so they are distinct from the preceding 21 and hence $V_\varepsilon$ has at least 24 non-degenerate critical points for all sufficiently small $\varepsilon >0$.
		
		Fix such an $\varepsilon > 0$.
		We have not eliminated the possibility that $V_\varepsilon$ has degenerate critical points elsewhere.
		To get rid of them, we perturb the vector $\mathbf{q}_0 = (1,1,1,q_\varepsilon,q_\varepsilon)$ of charges of the particles.\footnote{Essentially the same argument shows that we need only perturb the charges of any 4 of these particles. Alternatively, a different application of parametric transversality shows that we need only perturb the location of a single particle \cite[Thm.~6.2]{morse1969critical}.}
		Let $V^{\mathbf{q}}_\varepsilon$ be the electrostatic potential generated by particles at the same locations $\mathbf{a}_1,\ldots, \mathbf{a}_5$ as those generating $V_\varepsilon = V^{\mathbf{q}_0}_\varepsilon$, but with charge vector $\mathbf{q}$.
		Consider the smooth map
		\begin{equation*}
			\mathbf{F}\colon \R^5 \times (\R^3 \setminus \{\mathbf{a}_1,\ldots ,\mathbf{a}_5\}) \to \R^3, \qquad   \mathbf{F}(\mathbf{q},\mathbf{x}) =  -\nabla_{\mathbf{x}} V^{\mathbf{q}}_\varepsilon(\mathbf{x})  
		\end{equation*}
		given by the resulting electric field.
		Since $\mathbf{a}_1,\ldots, \mathbf{a}_5$ do not lie in a common plane, it is straightforward to show that $\mathbf{F}$ is a submersion (cf. \cite[Ex.~1.7.22]{guillemin1974differential}), so the parametric transversality theorem implies that $\mathbf{0}$ is a regular value for $\mathbf{F}(\mathbf{q},\cdot)$ and hence $V^{\mathbf{q}}_\varepsilon$ is a Morse function for almost every $\mathbf{q}\in \R^5$ \cite[p.~68]{guillemin1974differential}. 
		
		By the implicit function theorem, choosing any such $\mathbf{q}$ sufficiently close to $\mathbf{q}_0$ yields a Morse potential $V^{\mathbf{q}}_\varepsilon$ having at least 24 critical points.
		And since all of the charges are positive, all critical points of $V^{\mathbf{q}}_\varepsilon$ belong to the compact convex hull of $\{\mathbf{a}_1,\ldots,\mathbf{a}_5\}$, so there can be only finitely many.
	\end{proof}
	
	\begin{rmk}
		The construction employed in this counterexample can be iterated in a manner similar to constructions in \cite{edelsbruner2026counting}. Suppose that, after finitely many steps, the resulting potential $U(\mathbf x)$ enjoys the same \(\mathbf{D}_3\times\mathbf{Z}_2\)-symmetry as the original triangle configuration. The symmetry forces the Taylor expansion of $U(\mathbf x)$ around the origin to be of the form
		\begin{equation}\label{eq:Upotential}
			U(\mathbf x) = U(\mathbf 0) + \alpha H_2(\mathbf x) + \beta H_3(\mathbf x) + O(\lVert \mathbf x \rVert^4),
		\end{equation}
		where, for induction, we suppose \(\alpha\) and \(\beta\) are both greater than zero (as is the case for the original potential \(V_\triangle\) of the triangular configuration). Place a pair of charges of strength $q_\eta = \alpha \eta^3 - \gamma \eta^5$ at a distance $\pm \eta$ along the $z$-axis. The contribution of this pair near the origin is
		\begin{equation*}
			-\frac{q_\eta}{\eta^3}H_2 + \frac{q_\eta}{4\eta^5}H_4 + O \left(\frac{q_\eta \lVert \mathbf x \rVert^6}{\eta^7}\right),
		\end{equation*}
		so its leading $H_2$ term cancels the term \(\alpha H_2\) in \(U\). On the scale $\mathbf x = \lambda \eta^2 \mathbf X,$ and after dividing by $\eta^6,$ a repeat of the argument above shows that the new, deformed, potential converges to
		\begin{equation*}
			\gamma \lambda^2 H_2 + \beta \lambda^3 H_3 + \tfrac{1}{4} \alpha \lambda^4 H_4
		\end{equation*}
		as $\eta \to 0.$ For an appropriate choice of $\lambda$ and $\gamma,$ this potential has the same $21$ nondegenerate critical points as found above. For example, the choice $\lambda=\frac{2\beta}{5\alpha}, \gamma=\frac{\beta^2}{30\alpha}$ makes the limit of the deformed potential a multiple of the $\Phi_0$ above. One of these critical points is the origin, which was already a critical point before the new pair of charges was added. The remaining 20 are new. Since the potential of the added pair tends to zero in $C^2$ on compact sets disjoint from the origin, the nondegenerate critical points of $U$ persist in the new potential. Thus, we have added 20 critical points at the expense of adding two small positive charges. Moreover, after this step the expansion (\ref{eq:Upotential}) of the new potential retains the same form, with coefficient of $H_2$ equal to $\gamma \eta^2 > 0$ and coefficient of $H_3$ unchanged, so the argument may be repeated. Applying this iteration to the initial potential $V_{\triangle}$ and using the perturbation argument from the proof of Theorem \ref{thm:maintheorem} gives the following
	\end{rmk}
	
	\begin{propn}
		For every $m \geq 0,$ there exists a configuration of $3 + 2m$ positive point charges whose potential has a finite number $k \geq 4+20m$ of critical points, all of which are non-degenerate.
	\end{propn}
	
	Therefore, this construction produces an asymptotic critical point to charge ratio of $10$, which is larger than the value of $25/7$ achieved in \cite{edelsbruner2026counting}.

	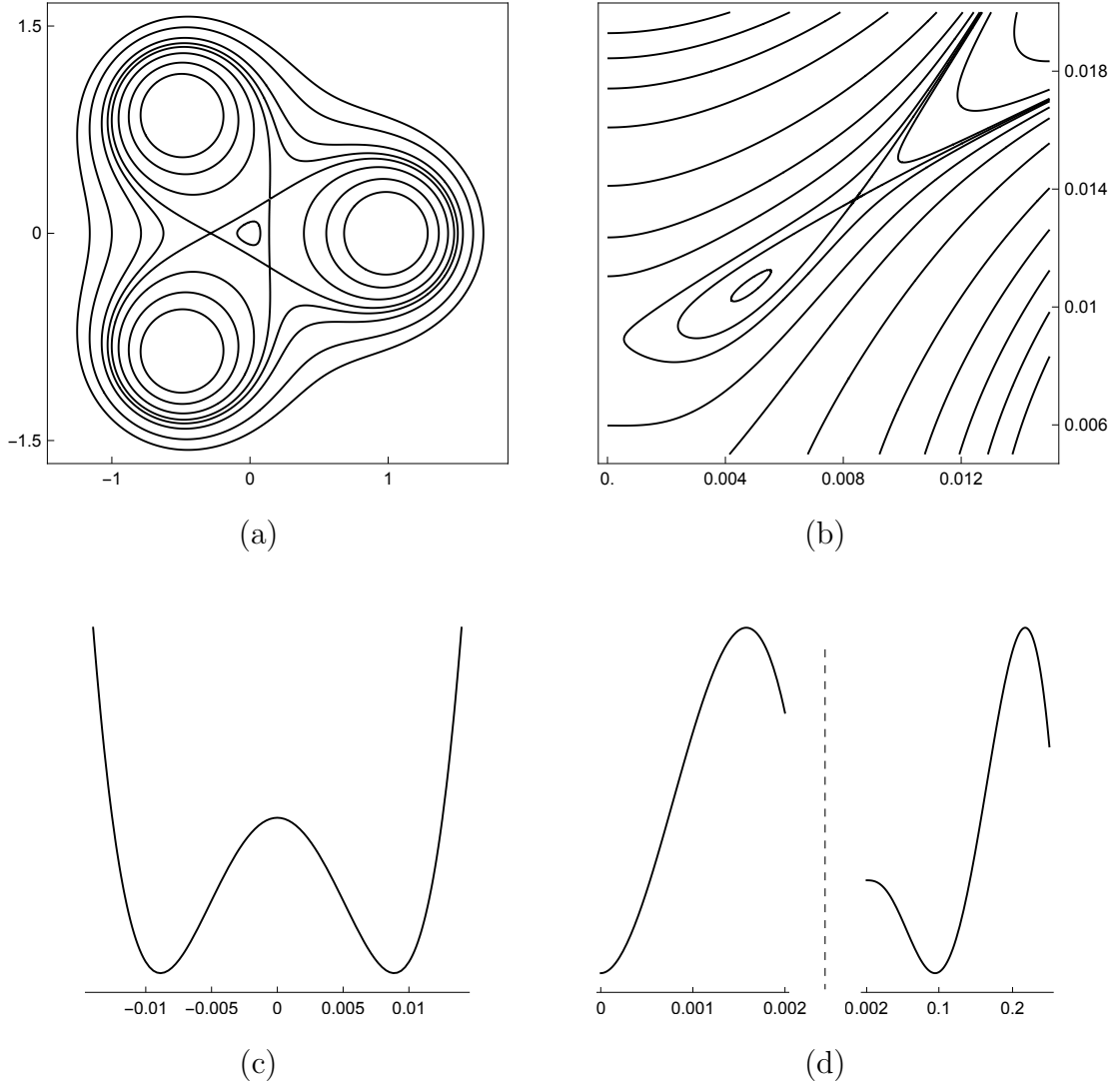
\begin{figure}
		\hspace*{-0.05\textwidth}
		\centering
		\begin{tikzpicture}
			\draw (-4, 4) node[inner sep=0] {\includegraphics[scale=1]{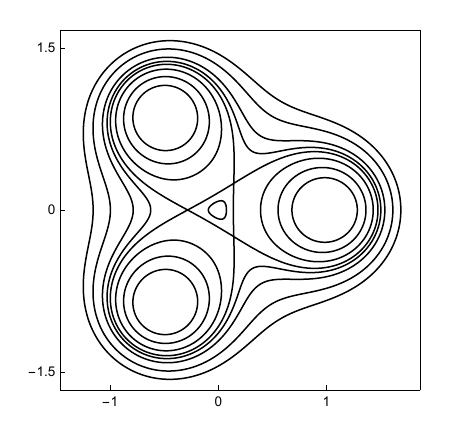}};
			\draw (-4,0) node {(a)};
			
			\draw (3.8,4) node[inner sep=0] {\includegraphics[scale=1]{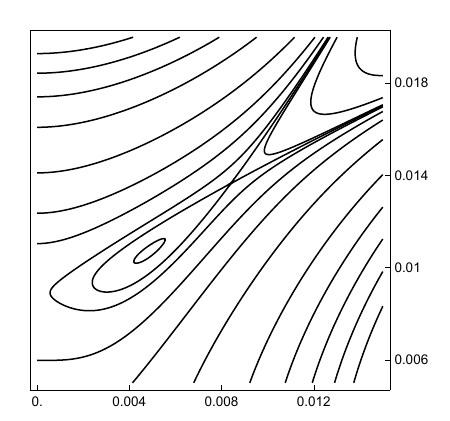}};
			\draw (3.5,0) node {(b)};

			\draw (-3.75,-3.5) node[inner sep=0]
			{\includegraphics[scale=1]{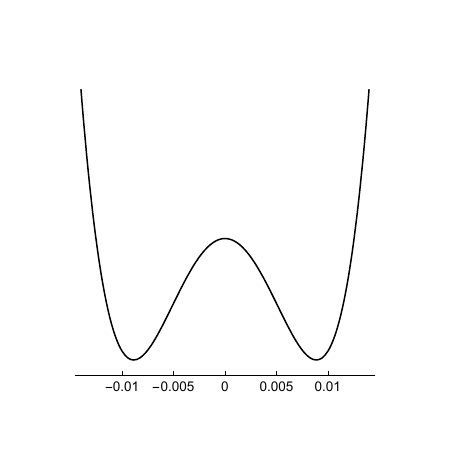}};
			\draw (-4,-7) node {(c)};

			\draw (1.75,-3.5) node[inner sep=0] {\includegraphics[scale=1]{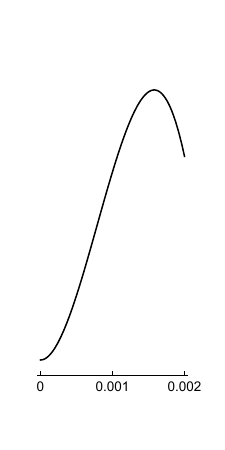}};
			%\draw (2.8,-5) node {$\lambda>I_3,~x_3<0$};
			
			\draw (5.25,-3.5) node[inner sep=0] {\includegraphics[scale=1]{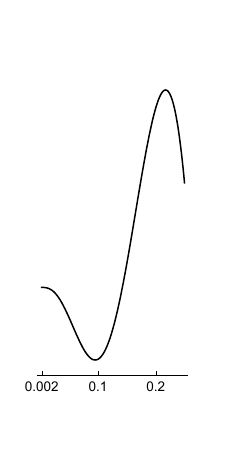}};
			\draw[dashed] (3.5,-1.5) -- (3.5,-6);
			
			\draw (3.5,-7) node {(d)};
			
		\end{tikzpicture}
		\caption{(a) The electric field generated by an equilateral triangle with no axial charges. The next three figures give numerical evidence for 24 equilibria when two axial charges are added with $\varepsilon=1/6$: (b) two off-planar equilibria in the positive $(x,z)$-quadrant of the  $y=0$ plane; (c) the potential along the $z$-axis showing 3 axial equilibria; (d) the potential along the ray  $z=0, \theta=\pi$ (with split axes) showing 3 planar equilibria.}
	\end{figure}

	\bibliographystyle{siam}
	\bibliography{crit_refs}

\end{document}